\documentclass[a4paper,UKenglish,cleveref, autoref, thm-restate]{socg-lipics-v2021}
\usepackage{tcolorbox}
\tcbuselibrary{skins,breakable}
\tcbset{enhanced jigsaw}

\newcommand{\eps}{\ensuremath{\varepsilon}}

\newcommand{\norm}[1]{\ensuremath{\|#1\|}}

\DeclareMathOperator*{\Prob}{\ensuremath{\textnormal{Pr}}}
\renewcommand{\Pr}{\Prob}

\newenvironment{tbox}{\begin{tcolorbox}[
		enlarge top by=5pt,
		enlarge bottom by=5pt,
		 breakable,
		 boxsep=0pt,
                  left=4pt,
                  right=4pt,
                  top=10pt,
                  arc=0pt,
                  boxrule=1pt,toprule=1pt,
                  colback=white
                  ]%%
	}
{\end{tcolorbox}}

\newcommand{\myqed}[1]{\let\qed\relax \hspace*{\fill} #1 \ensuremath{\square}}
\DeclareMathOperator{\arccosh}{arccosh}

\title{Locality Sensitive Hashing in Hyperbolic Space} %TODO Please add

\author{Chengyuan Deng}{Department of Computer Science, Rutgers University, USA}{cd751@rutgers.edu}{}{}
\author{Jie Gao}{Department of Computer Science, Rutgers University, USA}{jg1555@rutgers.edu}{[0000-0001-5083-6082]}{}
\author{Kevin Lu}{Department of Mathematics, Rutgers University, USA}{kll160@math.rutgers.edu}{}{}
\author{Feng Luo}{Department of Mathematics, Rutgers University, USA}{fluo@math.rutgers.edu}{}{}
\author{Cheng Xin}{Department of Computer Science, Rutgers University, USA}{cx122@rutgers.edu}{}{}

\authorrunning{C. Deng, J. Gao, K. Lu, F. Luo, C. Xin} %TODO mandatory. First: Use abbreviated first/middle names. Second (only in severe cases): Use first author plus 'et al.'

\Copyright{Chengyuan Deng, Jie Gao, Kevin Lu, Feng Luo and Cheng Xin} %TODO mandatory, please use full first names. LIPIcs license is "CC-BY";  http://creativecommons.org/licenses/by/3.0/


 %TODO mandatory; please add comma-separated list of keywords

\category{} %optional, e.g. invited paper

\relatedversiondetails{Full Version}{URL to related version} 
\funding{The authors would like to acknowledge funding support from NSF through CCF-2118953, DMS-2311064, DMS-2220271, IIS-2229876, CNS-2515159.}%optional, to capture a funding statement, which applies to all authors. Please enter author specific funding statements as fifth argument of the \author macro.

\nolinenumbers %uncomment to disable line numbering

\EventEditors{Hee-Kap Ahn, Michael Hoffmann, and Amir Nayyeri} \EventNoEds{3} \EventLongTitle{42nd International Symposium on Computational Geometry (SoCG 2026)} \EventShortTitle{SoCG 2026} \EventAcronym{SoCG} \EventYear{2026} \EventDate{June 2--5, 2026} \EventLocation{New Brunswick, NJ, USA} \EventLogo{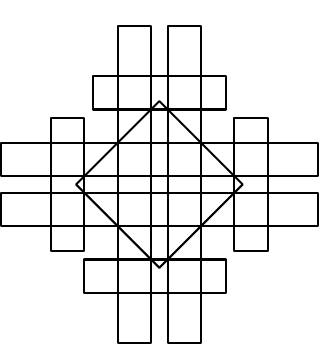} \SeriesVolume{367}
\ArticleNo{XX}     % <-- This will be filled in by the typesetters

\begin{document}

\maketitle
\begin{abstract}
For a metric space $(X, d)$, a family $\mathcal{H}$ of locality sensitive hash functions is called $(r, cr, p_1, p_2)$ sensitive if a randomly chosen function $h\in \mathcal{H}$ has probability at least $p_1$ (at most $p_2$) to map any $a, b\in X$ in the same hash bucket if $d(a, b)\leq r$ (or $d(a, b)\geq cr$). 
Locality Sensitive Hashing (LSH) is one of the most popular techniques for approximate nearest-neighbor search in high-dimensional spaces, and has been studied extensively for Hamming, Euclidean, and spherical geometries. 
An $(r, cr, p_1, p_2)$-sensitive hash function enables approximate nearest neighbor search (i.e., returning a point within distance $cr$ from a query $q$ if there exists a point within distance $r$ from $q$) with space $O(n^{1+\rho})$ and query time $O(n^{\rho})$ where $\rho=\frac{\log 1/p_1}{\log 1/p_2}$.
But LSH for hyperbolic spaces $\mathbb{H}^d$ remains largely unexplored.
%While tight bounds for $\rho$ are established for several metric spaces, such as Hamming, Euclidean, and spherical geometries, the algorithmic landscape for hyperbolic space $\mathbb{H}^d$ remains largely unexplored. 
In this work, we present the first LSH construction native to hyperbolic space. 
%The performance of LSH is characterized by the parameter $\rho$ that can be achieved as a function of the approximation factor $c$ of the nearest neighbor. 
For the hyperbolic plane $(d=2)$, we show a construction achieving $\rho \leq 1/c$, based on the hyperplane rounding scheme. For general hyperbolic spaces $(d \geq 3)$, we use dimension reduction from $\mathbb{H}^d$ to $\mathbb{H}^2$ and the 2D hyperbolic LSH to get $\rho \leq 1.59/c$. On the lower bound side, we show that the lower bound on $\rho$ of Euclidean LSH extends to the hyperbolic setting via local isometry, therefore giving $\rho \geq 1/c^2$.
\end{abstract}

\section{Introduction}
Locality Sensitive Hashing (LSH) is a hashing scheme in which similar items are more likely to hash to the same bucket. 
%is a fundamental algorithmic technique widely recognized across computational geometry and high-dimensional data analysis. 
Indyk and Motwani, in their seminal work~\cite{indyk1998approximate,v008a014}, developed LSH for general metric spaces and provided hash function constructions for the Hamming distance. The LSH framework provides an efficient data structure for the Approximate Nearest Neighbor Search (ANNS) problem, which is particularly challenging in the high-dimensional regimes. Since then, LSH has been developed for various similarity measures such as $\ell_1$~\cite{Andoni2006-vu, Linial2002-tn}, $\ell_2$~\cite{datar2004locality,andoni2008near}, $\ell_p$~\cite{datar2004locality}, cosine similarity~\cite{charikar2002similarity}, Jaccard similarity~\cite{Broder1997-ii, Broder2002-yj}, angular (spherical) distances~\cite{Terasawa2007-rw}, etc. 
With support for similarity search, data independence, mild dependence on data dimension, and simplicity in implementation, LSH serves as an algorithmic tool for classical problems such as clustering, near-duplicate detection, and efficient information retrieval for very large datasets. LSH has also been used widely in database and machine learning pipelines in industry~\cite{gionis1999similarity, buhler2001efficient, cohen2002finding, shakhnarovich2003fast, ravichandran2005randomized, das2007google,Koga2007-rs,Andoni2014-fj,Andoni2015-ur}. For a nice survey, please refer to~\cite{andoni2008near}. 

% \textcolor{red}{Feng: please feel free to modify and remove the following paragraphs. Indeed, Chengyuan is right.  Feel free to remove. Also, I wrote a few sentences at the beginning of Proof of Lemma B.2 in the appendix for integral geometry etc.  Feel free to remove those as well. }\cyd{This is really good, I just feel we do not need to say a lot about hyperbolic geometry, for a submission to SoCG.}\jie{I actually think it is OK to keep it here, not everyone is super familiar with hyperolic geometry in the SoCG community}

Hyperbolic geometry is a non-Euclidean geometry. It drops the parallel postulate of Euclidean geometry. In hyperbolic geometry, for any line and a point not on that line, there are at least two distinct lines through the point that do not intersect the original line.
Hyperbolic geometry has received increasing attention in modern data analysis because it naturally models hierarchical, tree-like, and exponentially branching structures that arise throughout science and technology. Unlike Euclidean spaces, where volume grows polynomially, hyperbolic spaces exhibit exponential volume growth, matching the combinatorial growth of hierarchies, taxonomies, knowledge graphs, and high-dimensional networks with latent tree-like structure. This geometric compatibility yields embeddings with dramatically reduced distortion: objects that are far apart in a hierarchy can be faithfully separated, and geodesics efficiently capture ancestry, flow, and community structure. As a result, hyperbolic embeddings now underpin state-of-the-art methods in network representation learning, hierarchical clustering, and discrete geometric data analysis~\cite{nickel17poincare,gulcehre2018hyperbolic,Ganea2018hyperbolic,tifrea2018poincare,De-Sa2018-cc,Chami2020-ug,Chami2019-ie, Wilson2014-qx, Karliga2004-fr, Tabaghi2024-bn}.  It is also the most important space representing the negative curvature phenomena arising from complex geometric situations~\cite{krioukov_hyperbolic_2010,Shavitt2004-au}.  Namely, according to Gromov’s hyperbolization program \cite{gromov}, spaces with complicated topology tend to resemble spaces with negative curvature. Hyperbolic geometry provides the fundamental model for negatively curved spaces, offering the canonical framework in which their curvature and geometric behavior are understood.

Moreover, the hyperbolic plane offers an unusually tractable analytic structure: its constant negative curvature provides explicit formulas for distance, geodesics, convexity, and isometries, enabling efficient optimization for machine-learning tasks. Its rich isometry group and uniformization properties are directly linked to complex analysis and conformal representations, providing tools that scale from theoretical guarantees to practical algorithms for dimensionality reduction, visualization, and manifold learning. 

%For these reasons, Hyperbolic space has recently been adopted in representation learning recently~\cite{nickel17poincare,gulcehre2018hyperbolic,Ganea2018hyperbolic,tifrea2018poincare,De-Sa2018-cc,Chami2020-ug,Chami2019-ie, Wilson2014-qx, Karliga2004-fr, Tabaghi2024-bn}.
%This is a particularly good fit when the input data has exponential growth rate -- a setting commonly seen in large scale networks. This is witnessed by the development of mathematical models  in hyperbolic space for complex networks~\cite{krioukov_hyperbolic_2010} as well as the adoption of hyperbolic embedding of large scale social networks and knowledge graphs~\cite{Chami2019-ie,Chami2020-ug}.
%https://arxiv.org/abs/1006.5169
These applications provide strong motivation for developing efficient algorithms in hyperbolic space, and in this paper, we study locality-sensitive hashing in hyperbolic space, a topic that, to our best knowledge, has not been studied before.

\subsection{LSH Framework}

We first recall the formal definitions on LSH and the Approximate Nearest Neighbor Search (ANNS) problem.

\begin{definition}[Locality Sensitive Hashing]
    Consider a metric space $(X,d)$\footnote{LSH is well-defined with a general dissimilarity measure, i.e., the triangle inequality does not necessarily hold. However, the approximate nearest neighbor search problem is typically studied in a metric space. Thus, we define LSH within the same context.}, let $U$ be a countable set and $\mathcal H$ be a family of functions from $X$ to $U$. 
    $\mathcal H$ is an $(r,cr,p_1,p_2)$-sensitive LSH family, where $r>0, c>1$, $0<p_1<p_2<1$ are real numbers, if for any two points $x,y \in X$ and a hash function $h$ chosen uniformly at random from $\mathcal H$ the following holds
\begin{align*}
d(x,y)\le r &\implies \Pr_{h\sim\mathcal H}[h(x)=h(y)]\ge p_1,\\
d(x,y)\ge cr &\implies \Pr_{h\sim\mathcal H}[h(x)=h(y)]\le p_2.
\end{align*}
\end{definition}

\begin{definition}[$(c, r)$-approximate Nearest Neighbor Search Problem]
    Let $P$ be a point set in metric space $(X,d)$. For any query point $x \in X$, if there exists a point $y \in P$ s.t. $d(x,y) \leq r$, then the algorithm returns $y' \in P$ s.t. $d(x,y') \leq c\cdot r$. 
\end{definition}
Previous works~\cite{indyk1998approximate, gionis1999similarity} have shown that the approximate nearest neighbor search problem can be solved by LSH, as formally stated below.
\begin{proposition}
    Given an $(r,cr,p_1,p_2)$-sensitive family $\mathcal{H}$, the $(c, r)$-ANNS problem can be solved using $O(n^{1+\rho})$ space, 
    %where $d$ is the dimension of the input point set,
    and $O(n^\rho\log n)$ query time, where $\rho = \rho (\mathcal{H})= \frac{\log(1/p_1)}{\log (1/p_2)}$.
\end{proposition}
Since both the space and time complexity are determined by $\rho$, we regard $\rho$ as an important performance parameter of LSH families. Namely, we ask for the tight bounds of $\rho$ related to the approximation factor $c$ for a specific metric space.

The tradeoff between the approximation factor $c$ and the performance exponent $\rho$ is intrinsically tied to the geometry of the underlying metric space. Previous studies have established a few tight bounds, with $\rho = \Theta(1/c)$ for Hamming distances and $\rho= \Theta(1/c^2)$ for Euclidean distances~\cite{indyk1998approximate,datar2004locality,andoni2008near,motwani2006lower,o2014optimal}. For the unit sphere under angular distance, Charikar~\cite{charikar2002similarity} proposed LSH based on hyperplane rounding schemes, which are refined by subsequent work~\cite{andoni2015practical} showing $\rho =\Theta(1/c^2)$. However, LSH remains underexplored for spaces of negative curvature, such as hyperbolic space. From a theoretical perspective, this gap is notable because hyperbolic geometry exhibits exponential volume expansion, which fundamentally distinguishes itself from the polynomial growth of Euclidean space. Therefore, even extending standard space-partitioning techniques such as quadtree~\cite{kisfaludibak24quadtree} to hyperbolic geometry is non-trivial.

Last, we remark that there has also been work on improved bounds (and empirical performance) for LSH in the data-dependent setting~\cite{Andoni2014-fj,Andoni2015-ur,andoni_et_al:LIPIcs.SoCG.2016.9,andoni2018data, andoni2018holder, Jayaram2024-qv}. The focus of this paper is on the data-independent case where the hash function is oblivious to input data $P$.

% This theoretical gap is notable given that hyperbolic geometry is the canonical domain for embedding hierarchical data and complex network structures, which can suffer high distortion when forced into Euclidean or Spherical metrics. 
\subsection{Our Results and Techniques}
In this paper, we give the first LSH construction in the hyperbolic space. We start with an LSH design for points in the hyperbolic plane, and then consider LSH families for high-dimensional hyperbolic spaces.

For points in $\mathbb{H}^2$, we are inspired by the  hyperplane rounding scheme by Charikar~\cite{charikar2002similarity} for cosine similarity between vectors. The hyperplane rounding scheme partitions the space by selecting a random hyperplane passing through the origin. Each data point is assigned a bit (0 or 1) determined by the sign of its projection onto the hyperplane's normal vector. In our construction, we partition the space using a random geodesic, which is a circular arc orthogonal to the unit circle in the Poincar\'e disk model. 

A technical caveat here is that we require a well-defined measure on the set of all geodesics. Such a measure is essential to ensure that random sampling of geodesics is meaningful and invariant under hyperbolic isometries. Fortunately, the classical integral geometry provides the needed theory. With some careful crafting using Crofton's formula, we are able to show that our construction yields $\rho \leq 1/c$ in $\mathbb{H}^2$.  This is formally stated as below.

\begin{theorem}
    \label{thm:lsh-ub-2d}
    For points in the hyperbolic plane $\mathbb{H}^2$ and $r>0, c>1$, there exists an $(r, cr, p_1, p_2)$-sensitive family $\mathcal{H}$, for some $0<p_2<p_1<1$, such that $\rho(\mathcal{H}) = \frac{\log (1/p_1)}{\log (1/p_2)} \leq 1/c$. 
\end{theorem}

To handle points in high-dimensional hyperbolic space, we use a random projection to low-dimensional hyperbolic space. First, we provide a refined analysis of the Johnson-Lindenstrauss type dimension-reduction result in~\cite{benjamini2009dimension}. Given a point set $P$ in a metric space $(X,d)$, we say a dimension reduction is an $(\alpha, \beta)$-approximation to dimension $k$, if there exists a mapping $f$ that projects $P$ into a $k$-dimension space $(\hat{X}, \hat{d})$, such that $\alpha \cdot d(p_i,p_j) \leq \hat{d}(f(p_i),f(p_j)) \leq \beta \cdot d(p_i,p_j)$ for any $p_i,p_j \in P$. In the literature, the embedding distortion is defined as $\beta/\alpha$.

\begin{theorem}[Informal Version of \Cref{thm:JLproof}]
    \label{thm:hyperbolic-dr}
    Let $k \in \mathbb{N}, k\geq 1$, if there exists an $(\alpha, \beta)$-approximate dimension reduction for an $n$-dimension Euclidean point set to $k$-dimension Euclidean space, then there exists an $(\alpha, \beta)$-approximate dimension reduction for a $(n+1)$-dimension hyperbolic point set to $(k+1)$-dimension hyperbolic space.
\end{theorem}

\Cref{thm:hyperbolic-dr} may be of independent interest. For example, an immediate corollary would be a Johnson-Lindenstrauss lemma for hyperbolic space. Namely, for $n$ points in a high-dimensional hyperbolic space, their pairwise distances can be preserved within $(1\pm \varepsilon)$-approximation while the dimension can be reduced to $O(\log n/\varepsilon^2)$, for $0<\varepsilon<1$.

In the context of LSH, our goal is to apply dimension reduction to obtain a projected point set in $\mathbb{H}^2$ and then employ our random geodesic partition technique. By \Cref{thm:hyperbolic-dr}, this task reduces to finding a Euclidean dimension reduction to a straight line $(k=1)$. Interestingly, this has been explored in the first Euclidean LSH construction~\cite{datar2004locality}, where they sample a projection matrix from the $p$-stable distribution (the Gaussian distribution for $p=2$) and project the input point set onto a straight line. We adopt this scheme and analyze the distance distortion in the hyperbolic plane. Formally, we obtain the following result for a general hyperbolic space.

\begin{theorem}
    \label{thm:lsh-ub-general-d}
    For points in $d$-dimensional hyperbolic space $\mathbb{H}^d$ with $d \geq 3$ and $r>0, c>1$, there exists an $(r, cr, p_1, p_2)$-sensitive family $\mathcal{H}$, for some $0<p_2<p_1<1$, such that $\rho(\mathcal{H}) = \frac{\log (1/p_1)}{\log (1/p_2)} \leq 1.59/c$. 
\end{theorem}

The existence of this family immediately implies that a sublinear query time of $\tilde{O}(n^{\rho})$ for Hyperbolic ANNS is possible, for $c>1.59$ in dimension $d\geq 3$. To the best of our knowledge, this is the first worst-case, data-oblivious sublinear query time guarantee for the ANNS problem in high-dimensional hyperbolic space. 

%\jie{discuss additional work on nearest neighbor search in hyperbolic space: \cite{Wu2020-cs} and graph based~\cite{prokhorenkova2022graphbased}.}

We remark that prior work~\cite{kisfaludibak24quadtree} has shown \textit{Locality Sensitive Ordering (LSO)} is not possible in hyperbolic space: no small collection of total orders can make all hyperbolic neighborhoods contiguous. This impossibility stems from the exponential volume growth of $\mathbb{H}^d$, which forces many mutually distant regions to interleave along a global ordering. Therefore, a global geometric linearization is not possible for hyperbolic space. On the other hand, our results confirm that this barrier does not apply to LSH, since hashing only requires distinguishing near and far points with good \emph{probability}, which can be achieved by the randomized geodesic-partition scheme.
% Our hyperbolic geodesic-partition scheme achieves this by without any global geometric linearization. 
% Thus, while deterministic ordering-based approaches fail in hyperbolic geometry, randomized partitioning remains viable.

On the lower bound side, it is straightforward to argue that the Euclidean LSH lower bound on $\rho$ extends to the hyperbolic setting. This follows from the fact that hyperbolic space is locally Euclidean; by scaling down a ``hard'' Euclidean dataset to a sufficiently small diameter, it can be embedded into hyperbolic space with arbitrarily low distortion. Therefore, we have the following result for the hyperbolic LSH lower bound.

\begin{restatable}{theorem}{LowerBound}
    \label{thm:lsh-lb}
    Fix $d \in \mathbb{N}$,  $c>1$, then there exists $0<\tau<1$, such that any $(\tau d,c\tau d,p_1,p_2)$-sensitive hashing family $\mathcal{H}$ in $\mathbb{H}^d$, for $0<p_2 <p_1< 1$, must satisfy $\rho(\mathcal{H}) \geq 1/c^2-o_d(1)$.
\end{restatable}

Our results leave a gap on $\rho$ between $O(1/c)$ and $\Omega(1/c^2)$. Closing this gap remains a major open problem on this topic.

\textbf{Experiments.} We show the performance of the proposed LSH constructions in both the hyperbolic plane and high-dimensional spaces. We observe that the value of $\rho$ for a random point set is a lot lower than $1/c$, which suggests the potential of the hyperbolic LSH in real-world applications.

\subsection{Related Work on ANNS in Hyperbolic Spaces}

Nearest neighbor query against a given set of $n$ points in the hyperbolic plane can be solved by using a standard point location query in the hyperbolic Voronoi diagram (that can be efficiently constructed \cite{Nielsen2010-vg}). With linear storage, a nearest-neighbor query can be performed in $O(\log n)$ time. This method is limited to 2D hyperbolic space. For $n$ points in $\mathbb{H}^d$ for a constant $d$, a $(1+\varepsilon)$-approximate nearest neighbor search can be done in time $O(d^{O(d)}\log n\log(1/\varepsilon)/\varepsilon^d)$ with space $O(nd^{O(d)}\log(1/\varepsilon)/\varepsilon^d)$~\cite{kisfaludibak24quadtree}. The performance deteriorates as $d$ increases.
Similar to the Euclidean setting, LSH is typically applied in high-dimensional spaces.

Krauthgamer and Lee~\cite{Krauthgamer2006-zz} studied approximate nearest neighbor in locally
doubling, $\delta$-hyperbolic spaces~\cite{gromov}. Specifically, they considered a $\delta$-hyperbolic metric space that has local geometry of type $(\lambda, \delta/3)$ in which every ball of radius $\delta/3$ in $X$ can be covered by $\lambda$
balls of half the radius. With $O(n^2)$ storage and query time $\lambda^{O(1)}\log^2n$, an approximate nearest neighbor with an $O(\delta)$ additive error is returned. Note that this result relies on a bounded doubling parameter $\lambda$, which is not necessarily true for a general point set in hyperbolic space $\mathbb{H}^d$. 

On the practical side, Wu and Charikar~\cite{Wu2020-cs} used Euclidean nearest neighbor search oracle as a black box for ANNS in hyperbolic spaces. Prokhorenkova et. al.~\cite{prokhorenkova2022graphbased} used graph-based nearest neighbor search on a similarity graph and analyzed query time assuming that the points are uniformly distributed within a ball of radius $R$. Another relevant result is by~\cite{andoni2018data}, which shows a data-dependent LSH for general metric space. In their construction, the LSH performance depends on the \emph{cutting modulus} of the metric space.

% Upper bounds for Euclidean LSH: \cite{datar2004locality, andoni2008near}.

% Lower bounds for Euclidean LSH: \cite{motwani2006lower, o2014optimal}.

% The many applications: \cite{gionis1999similarity, buhler2001efficient, cohen2002finding, shakhnarovich2003fast,  ravichandran2005randomized, das2007google}

% Recent developments on Data-dependent LSH: \cite{andoni2014beyond, andoni2015tight}.

% More to go.

% Recently there has been increasing interest in representation learning in non-Euclidean spaces, e.g., learning modules in hyperbolic ~\cite{nickel17poincare,gulcehre2018hyperbolic,Ganea2018hyperbolic,tifrea2018poincare,De-Sa2018-cc,Chami2020-ug,Chami2019-ie}.

\section{Locality Sensitive Hashing for $\mathbb{H}^2$}
\label{sec:H2}

In this section, we show the LSH construction in $\mathbb{H}^2$ that yields \Cref{thm:lsh-ub-2d}. We adopt the Poincar\'{e} disk model for points in $\mathbb{H}^2$, defined on the unit disk $\mathbb{D} = \{z \in \mathbb{C} : |z| < 1\}$. The space is equipped with the metric $ds = \frac{2|dz|}{1-|z|^2}$, which induces the following distance function:
\[d_\mathbb{H}(u, v) = \operatorname{arccosh}\left( 1 + 2\frac{|u-v|^2}{(1-|u|^2)(1-|v|^2)} \right).
\]
Geometrically, geodesics in $\mathbb{D}$ correspond to circular arcs orthogonal to the unit circle boundary  $C_{\infty}$, with diameters passing through the origin as special cases.

%Our construction exploits this model, where random geodesics can be sampled simply by choosing random boundary points, giving us a convenient way of using these random geodesics as the hyperbolic analog of separating hyperplanes.  

The locality sensitive hash functions we use are \emph{random geodesics} in the Poincar\'{e} disk model. 
For this purpose we require a well-defined probability measure on the set of all geodesics to ensure that random sampling of geodesics is meaningful and invariant under hyperbolic isometries. 
Luckily, the classical integral geometry~\cite{Santalo2004-sx} provides the needed theory.  Integral geometry studies the measures on the spaces of all lines, or totally geodesic planes in Euclidean, hyperbolic, or spherical spaces, such that the measures are invariant under isometries.    The fundamental theorems of Poincar\'e, Cauchy and Crofton show that one can recover basic geometric quantities such as lengths, areas, or volumes of an object by averaging over the set of geodesics intersecting the object. 

In the hyperbolic plane, the isometrically invariant probability measure on the space of all geodesics is given by the \emph{kinematic measure}~\cite{Santalo2004-sx}. Take a point $q$ and a non-zero tangent vector $v$ at $q$ in the hyperbolic plane.   The polar coordinate of a geodesic $\gamma$ with respect to $(q,v)$ is a pair $(t, \theta)$ where $t$ is the hyperbolic distance from $\gamma$ to $q$ and $\theta$ is the angle of the perpendicular from $q$ to $\gamma$ with the tangent vector $v$.
%a fixed direction. a geodesic can be defined by using geodesic polar coordinate $(t, \theta)$ where $t$ is the hyperbolic distance of the geodesic to the origin and $\theta$ is the angle of the perpendicular from the origin to the geodesic with a fixed direction. 
The kinematic measure in geodesic polar coordinates $(t,\theta)$ is
\[
d\mu(t,\theta) = \cosh(t)\, dt \, d\theta.
\]
%where $t \geq 0$ is the hyperbolic radial distance and $\theta \in [0,2\pi)$ is the angular coordinate.
The remarkable feature of the above formula is that it is independent of the choices of $q$ and $v$.   % See Figure \ref{fig:largest-rho}. 
This measure reflects the fact that the volume element in hyperbolic geometry grows exponentially with $t$. 
%Note that we overload the notation $\rho$ a little bit, because it is standard in both the literature of integral geometry and locality-sensitive hashing.

%Before detailing the construction, we address a measure-theoretic issue. 
Since the space of all geodesics in $\mathbb{H}^2$ has infinite total measure, we limit our choices of geodesics in the family of all geodesics intersecting a hyperbolic disk $B(0, R)$ of radius $R$ centered at the origin. By Crofton’s formula, the measure of the set of geodesics that intersect a convex body is proportional to the perimeter of that body. Specifically, the corresponding measure is obtained by integrating the invariant element $d\mu$ over the disk:
\begin{equation}\label{ball}
\int_{0}^{2\pi}\int_{0}^{R} \cosh(t)\, dt \, d\theta = 2\pi \sinh(R).
\end{equation}
%Thus, the measure of all geodesics intersecting a hyperbolic disk of radius $R$ is given by $2\pi \sinh(R)$.
%The crucial tool here is Crofton's formula, which establishes that the measure of the set of geodesics intersecting a convex body is proportional to the perimeter of that body. Therefore, to construct a valid LSH family, we restrict our sampling to the set of geodesics that intersect a compact bounding region containing the dataset $P$. 
Now, we choose $R$ such that the disk $B(0, R)$ centered at the origin with radius $R$ contains all points of $P$. Our final bound on the performance parameter $\rho$ does not depend on $R$, and in practice, one can choose $R$ to be a sufficiently large value.
Note that in practice, hyperbolic embedding typically avoids using points that are far from the origin due to practical concerns on resolution and numerical instability~\cite{pmlr-v202-mishne23a}. 

A geodesic $\gamma$ is chosen with probability density function $\frac{\cosh(t)}{2\pi\sinh(R)}$ to be at polar coordinate $(t, \theta)$ where $t$ is taken from $[0, R]$ and $\theta$ is taken from $[0,2\pi)$. In the Poincar\'{e} disk model, geodesics appear as Euclidean circular arcs, so $\gamma$ is a circular arc with Euclidean center $c \in \mathbb{R}^2$ and radius $b>0$.
Specifically, $c=\coth(t) e^{i\theta}$ (here convert $e^{i\theta}$ from the complex plane to the Euclidean plane) and $b=1/\sinh t$.
%\textcolor{red}{Feng: this formula seems t o be incorrect.  One reason is that it implies $|c|<1$. However, a circle perpendicular to the unit circle may have its center outside the unit ball.  My computation seems to say $c=\coth(t) e^{i\theta}$, and $b=1/\sinh(t)$. Kevin, could you double check of your formula? }\jie{You are right }
This defines a hash function $\mathbb H^2 \to \{-1, 1\}$ by 
\[h_\gamma(x) = \text{sgn}\left( \norm{x - c}^2 - b^2 \right).\]
where $x$ is the Euclidean coordinate of a point in $\mathbb{H}^2$. We can also rewrite the geodesic $\gamma$ using the Minkowski model with inner product $\langle z, z \rangle =-z_1^2+z_2^2+z_3^2$.  Let the geodesic $\gamma$ be given by $\{ x \in \mathbb H^2 |  \langle x, u\rangle =0\}$ for a point $u$ in the de Sitter space $\{ y \in \mathbb R^{2,1}| \langle y, y \rangle =1\}$: 
\[u=(\sinh t, \cosh t \cos\theta, \cosh t \sin \theta).\] %\jie{check the formula of $u$ in terms of $t, \theta$?  Feng: yes it is correct.} 
Note that $(t, \theta)$ is polar coordinate with respect to the point $q=(1,0,0)$ and $v =(0,1,0)$. Then the formula for $h_{\gamma}(x)$ is $\text{sgn} \langle x, u \rangle.$ This is exactly the same formula as in the Euclidean space. The advantage of this formula over the above one is that it is intrinsic to hyperbolic geometry, i.e., invariant under hyperbolic isometry. 

With this setup, we present the hash function in $\mathbb{H}^2$ below.

\begin{tbox}
    \textbf{Input:} A point set $P \subset \mathbb{H}^2$
    \begin{enumerate}
        \item Let $R$ be a radius such that $P$ is contained in the ball $B(0, R)$.
        \item Let $\mu_R$ denote the normalized kinematic measure restricted to the set of geodesics intersecting $B(0, R)$.
        \item Sample a random geodesic $\gamma \sim \mu_R$. 
        \item  Hash a point $x \in P$ to $h_\gamma(x)\in\{+1,-1\}$ according to which side of $\gamma$ it lies on (ties have measure zero).  
    \end{enumerate}
\end{tbox}

\begin{figure}[h]
    \centering
    \includegraphics[width=0.4\linewidth]{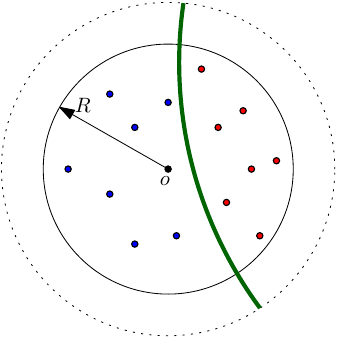}
    \caption{Illustration of the hashing scheme in $\mathbb H^2$. If two points are far, then they are more likely to be separated by a random geodesic (the green arc). Points in blue are assigned label $+1$ and points in red are assigned label $-1$.}
    \label{fig:hr-lsh}
\end{figure}

For brevity, we write $h(x)$ for $h_\gamma(x)$ when it is clear from context. \Cref{fig:hr-lsh} demonstrates how the hash function works as a separator. 
%Note that in practice it is important to concatenate multiple $h$ as the final hashing then repeat, as this process boosts the probability $p_1$ and reduce $p_2$. But it does not affect the value of $\rho$. 

The key to proving \Cref{thm:lsh-ub-2d} is a characterization of the collision probability under our hashing scheme, given in the following lemma. 

\begin{lemma}
\label{lem:collision-prob-hr}
Given two points $x,y$ in $\mathbb H^2$ of distance $d_{\mathbb H}(x,y)=r$, the kinematic measure $\mu$ of the set of all geodesics separating $x, y$ is $2 r$. In particular, if $x, y \in B(0, R)$, the ball of radius $R$ is centered at $0$,  the collision probability with respect to the normalized kinematic measure $\mu_R$ is
\begin{equation}\label{eq:collision}
\Pr[h(x)=h(y)|d_{\mathbb{H}}(x,y) = r] = 1-\frac{r}{\pi\sinh(R)}.
\end{equation}
\end{lemma}
\begin{proof} We use the Poincar\'{e} disk model for $\mathbb H^2$ in the calculation below.
For two points $x,y$ in $\mathbb H^2$ of distance $d_{\mathbb H}(x,y)=r$, let $X$ be the set of all geodesics in $\mathbb H^2$ separating $x,y$ and $Y$ be the set of all geodesics intersecting $B(0, R)$. Note that when $x,y\in  B(0, R)$, $X \subset Y$, i.e., $X \cap Y =X$. This shows,
\begin{equation} \Pr[h(x) \neq h(y) | d_{\mathbb H}(x,y)=r] = \frac{\mu(X \cap Y)}{\mu(Y)} =\frac{ \mu(X)}{\mu(Y)}.
\end{equation}
By \Cref{ball}, $\mu(Y) =2\pi\sinh(R)$. Next we show $\mu(X)=2r$, where $d_{\mathbb H}(x,y)=r$. 

\begin{figure}
\centering\includegraphics[width=0.4\linewidth]{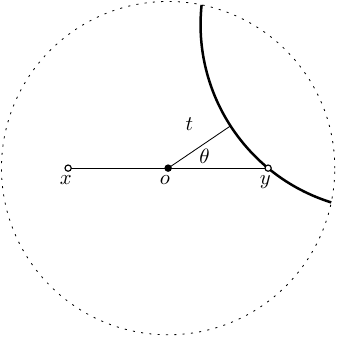}%{geo-sep-normalize.jpg}
    \caption{The geodesic intersects $y$, yielding the maximum value of $t$.}
    \label{fig:largest-rho}
\end{figure}
By the isometric invariance of the kinematic measure, to compute $\mu(X)$, 
we may assume, after applying a suitable isometry, that the midpoint of $x$ and $y$
is the origin and that both points lie on the 
$x$-axis.  Use the polar coordinate $(t, \theta)$ of a random geodesic $\gamma$ with respect to $(0, x\text{-axis})$. If $\gamma$ intersects the segment from $x$ to $y$, the furthest intersection is $x$ or $y$. \Cref{fig:largest-rho} illustrates this scenario. By the hyperbolic cosine law, we get
\[
\tanh(t) = \tanh(\frac{r}{2})\cos(\theta).
\]
It follows that $X$ in the polar coordinate is 
\[X=\{(t, \theta) | \theta\in [-\pi, \pi], 0\leq t \leq \text{arctanh}(\tanh(\frac{r}{2})\cos \theta )\}.\] Therefore,
\[\mu(X) = 2\int_{-\pi/2}^{\pi/2} \int_0^{\text{arctanh}(\tanh(\frac{r}{2})\cos\theta)} \cosh(t) \,dt \,d\theta.  \]
The double integral enjoys a very simple and elegant result, which is $r$ as shown in \Cref{clm:integral}. 
\begin{claim}
\label{clm:integral}
Let $r$ be a real number.
\[\int_{-\pi/2}^{\pi/2} \int_0^{\text{arctanh}(\tanh(r/2)\cos\theta)} \cosh(t) \,dt \,d\theta = r\]
\end{claim}

%Now for $x,y$ with distance $r$, by hyperbolic isometry we place $x,y$ symmetrically with midpoint at the origin. We compute the probability that the random geodesic $\gamma$ intersects the segment between $x,y$ (or the same, $\gamma$ separates $x,y$ in $\mathbb H^2$), i.e. $h(x) \neq h(y)$. \jie{I have one question here, we use isometry to assume $x, y$ have the origin as midpoint but with this isometry we are also changing the set of geodesics we sample from? (the disk of radius $R$ under the isometry is now some other disk?)}  \textcolor{red}{Feng: you are right. Since we did not normalize the measure $\mu_R$ in our previous version of the algorithm, the calculation (\ref{eq:collision}) was accurate with respect to the unnormalized measure. Now with the normalized measure $\mu_R$, we should be very careful.  Luckily, formula (\ref{eq:collision}) still represents the measure of the set of all lines separating $x,y$ and intersecting the ball $B$ of radius $R$.  The reason is that if $x, y$ are in $B$ and a line $L$ separates $x, y$, then $L$ intersects $B$ automatically, i.e., the location of $B$ is irrelavent as long as it contains $x,y$. }

%Let $t$ be the distance between the origin to the random geodesic $\gamma$. If $\gamma$ intersects the segment, the furthest intersection is $x$ or $y$. \Cref{fig:largest-rho} illustrates this scenario. By the hyperbolic cosine law, we get
%\[
%\tanh(t) = \tanh(\frac{r}{2})\cos(\theta)
%\]

Therefore, $\mu(X)=2r$. Furthermore, the probability that $x, y$ are separated by $\gamma$ when $\gamma$ is chosen from $Y$ is:
\[
\Pr[h(x) \neq h(y) | d_{\mathbb{H}}(x,y)=r] = \frac{2r}{2\pi\sinh(R)}%\frac{2\int_{-\pi/2}^{\pi/2} \int_0^{\text{arctanh}(\tanh(\frac{r}{2})\cos\theta)} \cosh(t) \,dt \,d\theta}{2\pi\sinh(R)}
\]
and the collision probability is:
\[
\Pr[h(x) = h(y) | d_{\mathbb{H}(x,y)}=r] =1-\frac{r}{\pi\sinh(R)}
\]
This finishes the proof.
\end{proof}
The detailed calculation of \Cref{clm:integral}  is deferred to \Cref{appendix:proof-H2}.
%, where we give more background and details on integral geometry and Crofton formula, and how we use them to derive the collision probability. 
For now we focus on the upper bound of $\rho$, note that the lemma immediately gives
\begin{align*}
    p_1 = 1-r/\pi\sinh(R),\: p_2 = 1-cr/\pi\sinh(R)
\end{align*}
In order to bound $\rho=\frac{\log 1/p_1}{\log 1/p_2}$, we apply Lemma 1 in~\cite{datar2004locality}, copied below.

\begin{lemma}[Lemma 1~\cite{datar2004locality}]\label{lem:log-ratio}
For $x\in [0, 1)$ and $\ell>1$ such that $1-\ell x>0$, \[\frac{\log (1-x)}{\log (1-\ell x)}\leq \frac{1}{\ell}.\]
\end{lemma}
This shows that $\rho\leq 1/c$.
%Combining \Cref{lem:collision-prob-hr} and \Cref{lem:mono-hr}, 
Thus we conclude the proof of \Cref{thm:lsh-ub-2d}.
\section{Dimension Reduction in Hyperbolic Space}\label{sec:dim-reduction}

In this section, we prove the dimension reduction result (\Cref{thm:JLproof}). We use the Poincar\'e half-space model of the hyperbolic space $\mathbb{H}^d$. Each point in 
$\mathbb{H}^d$ is represented by a pair $(z, x)$ with $z\in \mathbb{R}^+$ and $x\in \mathbb{R}^{d-1}$. The hyperbolic distance between two points $p_1=(z_1, x_1)$ and $p_2=(z_2, x_2)$ is defined by the following function
\[d_{\mathbb{H}}(p_1, p_2)=\arccosh \left (1+\frac{\|x_1-x_2\|^2+(z_1-z_2)^2}{2z_1z_2}\right ),\] where $\|x_1-x_2\|$ is the Euclidean distance between $x_1, x_2$.
For a set of $n$ points $P=\{p_i=(z_i, x_i)\}\subseteq \mathbb{H}^d$, we would like to perform dimension reduction to generate $P'=\{p'_i\}$ in $\mathbb{H}^k$ and keep the pairwise distances similar.  The method by Benjamini and Makarychev~\cite{benjamini2009dimension} is to perform dimension reduction for the Euclidean point set $X=\{x_i\} \subseteq \mathbb{R}^{d-1}$ by a function $f$ to $X'=\{x'_i\} \subseteq \mathbb{R}^{k-1}$ and define $P'=\{p'_i\}$ with $p'_i=(z_i, x'_i)$ in $\mathbb{H}^k$.
Below, we present a refined analysis that yields slightly tighter bounds on the distortion. Essentially, we show that the bound on the distortion of the Euclidean distances of $X'$ (compared to $X$) carries over to bound the distortion of the hyperbolic distance of $P'$ compared to $P$.

\begin{definition}
    Let  $f: (U, d_U) \to (V, d_V)$ be a map between two metric spaces.  If there are positive constants $\gamma_1, \gamma_2$ such that for any $u_1, u_2\in U$, \[\gamma_1 \cdot d_U(u_1, u_2)\leq d_V(f(u_1), f(u_2))\leq \gamma_2 \cdot d_U(u_1, u_2),\]
     we say that the distance stretch is upper bounded by $\gamma_2$ and lower bounded by $\gamma_1$ and the distortion is $\gamma_2/\gamma_1$.
\end{definition}

\begin{theorem}\label{thm:JLproof}
    Given $n$ points $P=\{p_i=(z_i, x_i)\}$ in $\mathbb{H}^d$. Suppose we have a function $f: \mathbb{R}^{d-1}\rightarrow \mathbb{R}^{k-1}$ for $n$ input points $X=\{x_i\}$ such that $\forall x_i, x_j\in X$, \[\gamma_1 \cdot \|x_i-x_j\|\leq \|f(x_i)- f(x_j)\|\leq \gamma_2 \cdot \|x_i- x_j\|,\] with $0<\gamma_1\leq 1$ and $\gamma_2\geq 1$. Then the stretch bounds and distortion of the map \[g(p)=g((z, x))=(z, f(x))\] are the same: $\forall p_i=(z_i, x_i), p_j=(z_j, x_j)\in P$,
    \[\gamma_1 \cdot d_{\mathbb{H}}(p_i, p_j)\leq d_{\mathbb{H}}(g(p_i), g(p_j))\leq \gamma_2 \cdot d_{\mathbb{H}}(p_i, p_j).\]
\end{theorem}

The proof follows the main idea in~\cite{benjamini2009dimension}, but we have slightly improved analysis and bounds.
%except that instead of using inequality $\tanh t \leq \frac{3t}{1+2t}$, we can use $\tanh t \leq t$, when $t>0$.  
Specifically, the main idea is to analyze the behavior of the function
\begin{equation}\label{eq:F}
F_{z_1, z_2}(r)=\arccosh \left (1+\frac{r^2+(z_1-z_2)^2}{2z_1z_2}\right),
\end{equation} 
which is an increasing function of $r$, when $z_1, z_2>0$. For completeness, we include the full proof in Appendix~\ref{sec:JLproof}.
\begin{lemma}\label{lem:F}
For $z_1, z_2>0$, \begin{enumerate}
    \item For $\gamma\geq 1$, $F_{z_1, z_2}(\gamma \cdot r)\leq \gamma \cdot F_{z_1, z_2}(r)$.
    \item For $0<\gamma'\leq 1$, 
    $F_{z_1, z_2}(\gamma' \cdot r)\geq \gamma' \cdot F_{z_1, z_2}(r)$.
\end{enumerate}
\end{lemma}

\begin{proof}[Proof of Theorem~\ref{thm:JLproof}]
$d_{\mathbb{H}}(g(p_i), g(p_j))=F_{z_i, z_j}(\|f(x_i)-f(x_j)\|)$.
By definition, $\gamma_1 \|x_i-x_j\|\leq \|f(x_i)-f(x_j)\|\leq \gamma_2 \|x_i-x_j\|$, with $0<\gamma_1\leq 1$ and $\gamma_2\geq 1$. 
The claim follows from $d_{\mathbb{H}}(g(p_i), g(p_j))=F_{z_i, z_j}(\|f(x_i)-f(x_j)\|)$, that $F_{z_i, z_j}(x)$ is a monotonically increasing function of $x$ and Lemma~\ref{lem:F}.
%Thus, for the upper bound, $\gamma_2\geq1$, we have $$d(g(p_i), g(p_j))=F_{z_i, z_j}(\|f(x_i)-f(x_j)\|)\leq F_{z_i, z_j}(\gamma_2 \cdot \|x_i-x_j\|) \leq \gamma_2 \cdot d(p_i,p_j).$$ The first inequality is due to that $F_{z_i, z_j}(r)$ is increasing with $r$. The second inequality follows from Lemma~\ref{lem:F}. For the lower bound side, $0<\gamma_1\leq1$, we have $$d(p_i,p_j)=F_{z_i, z_j}(\|x_i-x_j\|)\leq \frac{1}{\gamma_1} \cdot F_{z_i, z_j}(\gamma_1 \|x_i-x_j\|)\leq \frac{1}{\gamma_1} \cdot  d(g(p_i), g(p_j)).$$
%This finishes the proof.
\end{proof}

Specifically, one can apply Johnson Lindenstrauss Lemma~\cite{johnson1984extensions} on the Euclidean coordinates $\{x_i\}$ -- by random linear projection to dimension $O(\log n/\varepsilon^2)$ -- and arrive at a corresponding Johnson-Lindenstrauss Lemma in hyperbolic geometry. 

\begin{corollary}[Johnson Lindenstrauss Transform in $\mathbb{H}^d$]
For $n$ points $P=\{(z_i, x_i)\}$ in $\mathbb{H}^d$ we can project it to points $P'=\{(z_i, x'_i)\}$ in $\mathbb{H}^{k}$ with $k=\Theta(\log n/\varepsilon^2)+1$ and stretch upper and lower bounded by $1\pm \varepsilon$ respectively.
%-- by using Johnson Lindenstrauss Lemma to project $\{x_i\}$ in in $\mathbb{R}^{d-1}$ to $\{x'_i\}$ in $\mathbb{R}^{k-1}$ which achieves the stretch bounds.
\end{corollary}
\section{Locality Sensitive Hashing for $\mathbb{H}^d$}

A distribution $D$ is called $p$-stable~\cite{zolotarev86stable} with $p\geq 0$ if for any $n$ real numbers $v_1, v_2, \cdots, v_n$ and random independent variables $X_1, X_2, \cdots X_n$ with distribution $D$, $\sum_i v_iX_i$ has the same distribution as the variable $(\sum_i |v_i|^p)^{1/p}X$, where $X$ is a random variable with distribution $D$. 
A Gaussian distribution $\mathcal{N}(0, 1)$ is $2$-stable. This property has been used to build Euclidean LSH~\cite{datar2004locality}, to sketch high-dimensional vectors~\cite{Indyk2006-cv}, and for many other applications.

Consider $n$ points $P=\{p_i\}$ in $\mathbb{H}^d$. We use random projection to project these points to $P'=\{p'_i\}$ in $\mathbb{H}^2$ using the method in Section~\ref{sec:dim-reduction}. In particular, we take a vector $\vec{a}$ with dimension $d-1$ where each entry is independently taken from a Gaussian distribution $\mathcal{N}(0, 1)$. For $p_i=(z_i,x_i)$ in the upper half space model, we have $p_i=(z_i, x'_i)$ in $\mathbb{H}^2$, where $x'_i=\vec{a}\cdot x_i$. We then use the LSH mechanism in Section~\ref{sec:H2} to map these points into buckets. In this section, we analyze the performance of this mechanism. 

In our case,
the value $\vec{a}\cdot (p-q)$ for any two vectors $p, q\in \mathbb{R}^{d-1}$ is also a Gaussian distribution with zero mean and variance $\|p-q\|^2$, or, the distribution of $\|p-q\|Z$ with $Z\sim \mathcal{N}(0, 1)$. Take $f(z)$ to be the probability density function of the absolute value of $\mathcal{N}(0,1)$: \[f(z)=\frac{2}{\sqrt{2\pi}} e^{-z^2/2}.\] 

To calculate the probability that $p_1, p_2$ are mapped to the same bucket, we take $s=\|x_1-x_2\|$ and $r=d_{\mathbb{H}}(p_1, p_2)=F_{z_1, z_2}(s)$, where \[F_{z_1, z_2}(s)=\arccosh \left (1+\frac{s^2+(z_1-z_2)^2}{2z_1z_2}\right),\] is an increasing function of $s$. 
Suppose $p_1, p_2$ are mapped to points $p'_1, p'_2$ with $s'=\|x'_1-x'_2\|$ and $r'=d_{\mathbb{H}}(p'_1, p'_2)=F_{z_1, z_2}(s')$. $s'$ follows the distribution of $sZ$ with $Z\sim \mathcal{N}(0, 1)$. Recall that using the LSH for $\mathbb{H}^2$  the probability that $p'_1, p'_2$ map to the same bucket is $1-\frac{r'}{w}$ where $w=2\sinh R$.
Now we have,
\begin{align*}
p(r)=\text{Pr}(h(p_1)=h(p_2))&=\int_{z=0}^{\infty}f(z)\left (1-\frac{F_{z_1, z_2}(s\cdot z)}{w}\right )dz.
\end{align*}

We can get an upper bound of $p(r)$.
\begin{lemma}\label{lem:UB}
\begin{equation}\label{eq:UB}
p(r)\leq 1-\alpha \frac{r}{w},\, \alpha\approx 0.63.
\end{equation}
%\int_{z=0}^{\infty}f(z)\left (1-\frac{r\cdot z}{w}\right )dz=1-\frac{2r}{w\sqrt{2\pi}}$$
\end{lemma}

\begin{proof}
    We first rewrite:
\[ \int_{z=0}^{\infty}f(z)\left (1-\frac{F_{z_1, z_2}(s\cdot z)}{w}\right )dz = 1-\int_{z=0}^{1}f(z)\frac{F_{z_1, z_2}(s\cdot z)}{w}dz - \int_{z=1}^{\infty}f(z)\frac{F_{z_1, z_2}(s\cdot z)}{w}dz\]
For the first integral, by Lemma~\ref{lem:F}, $F_{z_1, z_2}(s\cdot z)\geq z \cdot F_{z_1, z_2}(s)=zr$ in the range of $0<z\leq 1$.
\[\int_{z=0}^{1}f(z)\frac{F_{z_1, z_2}(s\cdot z)}{w}dz \geq \int_{z=0}^{1}f(z)z \frac{r}{w}dz =\sqrt{\frac{2}{\pi}}(1-e^{-1/2})\frac{r}{w}\approx .31\frac{r}{w}\]
For the second integral, $z\geq 1$, $F_{z_1, z_2}(s\cdot z)\geq F_{z_1, z_2}(s)=r$.
\[\int_{z=1}^{\infty}f(z)\frac{F_{z_1, z_2}(s\cdot z)}{w}dz \geq \int_{z=1}^{\infty}f(z)\frac{r}{w}dz \approx .32 \frac{r}{w}\]
Adding them together gives us:
\[p(z) \leq 1-.63\frac{r}{w}\]
\end{proof}

We also have a lower bound of $p(r)$.
\begin{lemma}\label{lem:LB}
\begin{equation}\label{eq:LB}
p(r)\geq 1-\frac{r}{w}.
\end{equation}
\end{lemma}
\begin{proof}
Here, $z$ follows the distribution of the absolute value of a Gaussian distribution $\mathcal{N}(0, 1)$, $z\geq 0$, $\mathbb{E}
[z]=\sqrt{2/\pi}$ and $\mathbb{E}[z^2]=1$.

The function $\arccosh(x)$ with $x>1$ is concave and monotonically increasing. 
For a concave function $b(x)$ one can write it as \[b(x)\leq b(x_0)+b'(x_0)(x-x_0).\]
Applying this for $b(x)=\arccosh(x)$ at $x=1+\frac{s^2z^2+(z_1-z_2)^2}{2z_1z_2}$ and $x_0=1+\frac{s^2+(z_1-z_2)^2}{2z_1z_2}$  we have
\[F_{z_1, z_2}(sz)\leq F_{z_1, z_2}(s)+b'(x_0)\cdot \frac{s^2(z^2-1)}{2z_1z_2}\]
Taking expectation over $z$ on both sides and recall $\mathbb{E}[z^2]=1$ we have 
\[\mathbb{E}[F_{z_1, z_2}(sz)]\leq F_{z_1, z_2}(s)+b'(x_0)\cdot \frac{s^2(\mathbb{E}[z^2]-1)}{2z_1z_2}=F_{z_1, z_2}(s)=r\]
Now we have \[p(r)=1-\mathbb{E}[F_{z_1, z_2}(sz)]/w\geq 1-r/w\]
This finishes the proof.
\end{proof}

%On the other hand, we have  $$1-\frac{F_{z_1, z_2}(s\cdot z)}{w}\leq 1-\frac{F_{z_1, z_2}(s)}{w}=1-\frac{r}{w}$$ when $z\geq 1$ and $$1-\frac{F_{z_1, z_2}(s\cdot z)}{w}\geq 1-\frac{F_{z_1, z_2}(s)}{w}=1-\frac{r}{w}$$ when $z\leq 1$. 

Thus $p_1=p(r)\geq 1-r/w$. $p_2=p(c\cdot r)\leq 1-\frac{
\alpha cr}{w}$. Since we need $p_1>p_2$, this requires that $c\geq 1/\alpha\approx 1.59$.

Now we can estimate the function $\rho=\frac{\log 1/p(r)}{\log 1/p(c\cdot r)}$. Again we apply \Cref{lem:log-ratio} and get \[\rho=\frac{\log 1/p(r)}{\log 1/p(c\cdot r)}\leq \frac{\log (1-\frac{r}{w})}{\log (1-\frac{\alpha cr}{w})}\leq \frac{1}{\alpha c}\approx \frac{1.59}{c}.\]
This concludes \Cref{thm:lsh-ub-general-d}.
%\jie{possibly add a theorem concluding LSH in $H^d$}
% \begin{theorem}
%     Given a point set $P \subseteq \mathbb{H}^d$ with $d \geq 3$, for $r>0, c>1, 0<p_2<p_1<1$, there exists an $(r, cr, p_1, p_2)$-sensitive family $\mathcal{H}$, such that $\rho(\mathcal{H}) = \frac{\log (1/p_1)}{\log (1/p_2)} \leq 1.59/c$. 
% \end{theorem}
\section{Lower Bound}
The lower bound on the performance parameter $\rho$ for LSH of $\ell_p$ metrics has been studied in~\cite{motwani2006lower,o2014optimal}. Theorem 5.2 in~\cite{o2014optimal} stated that for points of  the Hamming cube $\{0, 1\}^d$ and some $0<\delta<1$, an $(r, R, p_1, p_2)$-sensitive LSH of Hamming distance, with $r=\delta d/c$ and $R=\delta d$, must have \[\rho(c)=\frac{\log1/p_1}{\log1/p_2}\geq \frac{1}{c} -o_d(1)\]

This lower bound can be used to generate a lower bound of $1/c^p$ for $\ell_p$ distance -- by using $\ell_p$ distance on the same point set. In particular, it shows that LSH for Euclidean distances has $\rho\geq 1/c^2$. We will use the same construction to show a lower bound on $\rho$ for hyperbolic distance. 

\LowerBound*
\begin{proof}
For a point $x_i$ of the Hamming cube $\{0, 1\}^d$, we define $p_i\in \mathbb{H}^{d+1}$ to be $(z, x_i)$ in the half-space model, with a very large value of $z$ to be decided later. Now consider the hyperbolic distance of $p_i, p_j$: 
\[d_{\mathbb{H}}(p_i, p_j)=\arccosh \left(1+\frac{\|x_i-x_j\|^2}{2z^2}\right )\]
For any $0<\eps<1$, we take $z\geq \sqrt{d/2\eps}$. Thus  $\frac{\|x_i-x_j\|^2}{2z^2}\leq \eps$.
Use the generalized power series (Puiseux series) \[\arccosh(1+x)=\sqrt{2}\sum_{t=0}^{\infty}\frac{(-1)^t{2t \choose t}}{(2t+1)8^t}x^{t+1/2}=\sqrt{2x}\left(1-\frac{1}{12}x+\frac{3}{160}x^2 -\frac{5}{896}x^3+\cdots \right)\] 
%\jie{the above expansion formula is from Gemini -- not sure if we can trust it, someone should also verify}  \textcolor{red}{Feng: I checked it using Maple. The formula is correct.  On the other hand, what inequality for arccosh(1+x) are you using? We can probably give a direct proof.}\jie{I added the ineuqality below}
This gives us, for small positive $x$,  \[\sqrt{2x}(1-\frac{1}{12}x)\leq \arccosh(1+x)\leq \sqrt{2x}.\]
We can now bound 
\[\sqrt{\|x_i-x_j\|_1}  (1-\eps/12)\leq  d_{\mathbb{H}}(p_i, p_j)\cdot z\leq \sqrt{\|x_i-x_j\|_1}, \]
where the $\ell_1$ distance (same as the Hamming distance)  $\|x_i-x_j\|_1$ is the same as the squared $\ell_2$ distance $\|x_i-x_j\|_2^2$. 
%\cyd{This is not true right? $\|x_i-x_j\|_1=\|x_i-x_j\|^2$, so by the end you still get $1/c^2$}
Now we build an $(r_h, R_h, p_1, p_2)$-sensitive LSH such that $r_h\cdot z\leq \sqrt{r}$ and $\sqrt{R}(1-\eps/12)\leq R_h\cdot z$. Take $R_h=c_h\cdot r_h$, we have \[\frac{1}{c}=\frac{r}{R}\geq \left(\frac{r_h}{R_h}\right)^2(1-\eps/12)^2=\frac{1}{c_h^2}(1-\eps/12)^2\]
Combining everything, a LSH for $\mathbb{H}^{d+1}$ must have \[\rho(c_h)\geq \frac{1}{c}-o_d(1)\geq \frac{1}{c_h^2}(1-\eps/12)^2-o_d(1)=\Omega(1/c_h^2)\]
\end{proof}

\section{Experiments}

We implement our Locality Sensitive Hashing algorithms in both $\mathbb{H}^2$ and $\mathbb{H}^d$. We aim to evaluate the LSH performance in terms of $\rho$ in practice.

We first explain the synthetic data we use. For each $d \in \{2, 10, 100, 1000\}$, we generate a dataset of 1000 points by uniformly random sampling in a hyperbolic sphere of the same radius, $R = \log(199)$. This radius was chosen so that after the points are projected to $\mathbb{H}^2$ and mapped to the Poincare disk model, the points lie inside a Euclidean circle of radius 0.99. Further, we set $r = 0.2$ with varying $c = 1.5 + k$ for $k = 0,1 ... 17$ to compute the value of $p_1$ and $p_2$. We illustrate results on $p_1, p_2$ in \Cref{fig:p1p2} and $\rho$ in \Cref{fig:rho} for different choices of $c$ obtained by averaging over 1000 repetitions.

% We ran experiments by conducting LSH on data in $\mathbb{H}^d$ for $d = 2, 10, 100, 1000$. For each of the dimensions, we experimented with a set of points generated by uniformly randomly sampling 1000 points in a Hyperbolic sphere of the same radius, $\log(1.99/100)$. This radius was chosen so that after projecting down to $\mathbb{H}^2$ and mapped to the Poincare disk model, the points lied inside a circle of radius .99. After that mapping, we sampled geodesics randomly in the Poincare disk according to Crofton's formula with maximum displacement set to be .99 in Euclidean distance. Then, we placed the points into buckets according to which side of the geodesic they were on to calculate the average $p_1$ and $p_2$ probabilities. Specifically, $p_1$ was calculated by finding all pairs of points with distance less than $r = .2$, and dividing the pairs that ended up in the same bucket over all pairs. $p_2$ was found similarly by checking pairs of points with distance greater than $c\cdot r$, for different values of $c$. We did this 1000 times and averaged all of the resulting $p_1$, $p_2$ probabilities to get $p_1'$ and $p_2'$. Then, we calculate $\rho' = \frac{\log(1/p'_1)}{\log(1/p'_2)}$ from the resulting averages. We do this for each $c = 1.5 + k$ for $k = 0,1 ... 17$ which changes the $p'_2$ numbers. \Cref{fig:p1p2rho} displays the results.
\begin{figure}
    \centering
    \includegraphics[width=0.48\linewidth]{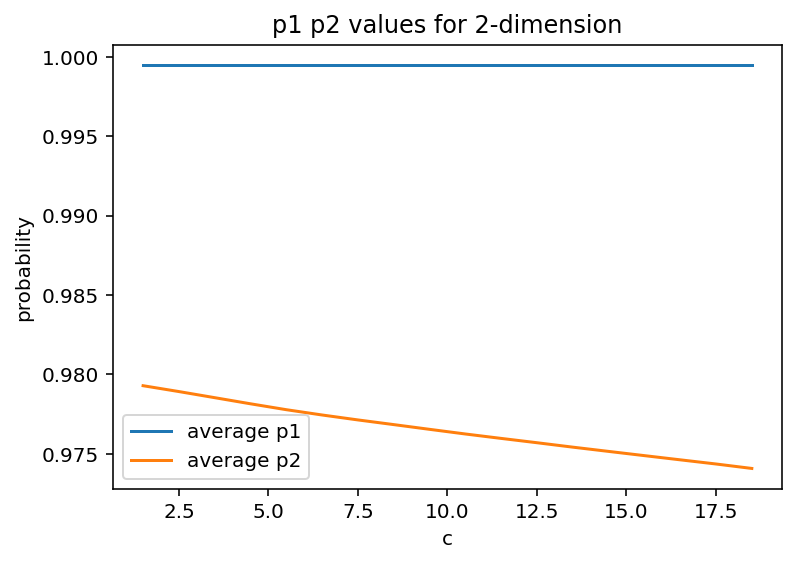}
    \includegraphics[width=0.48\linewidth]{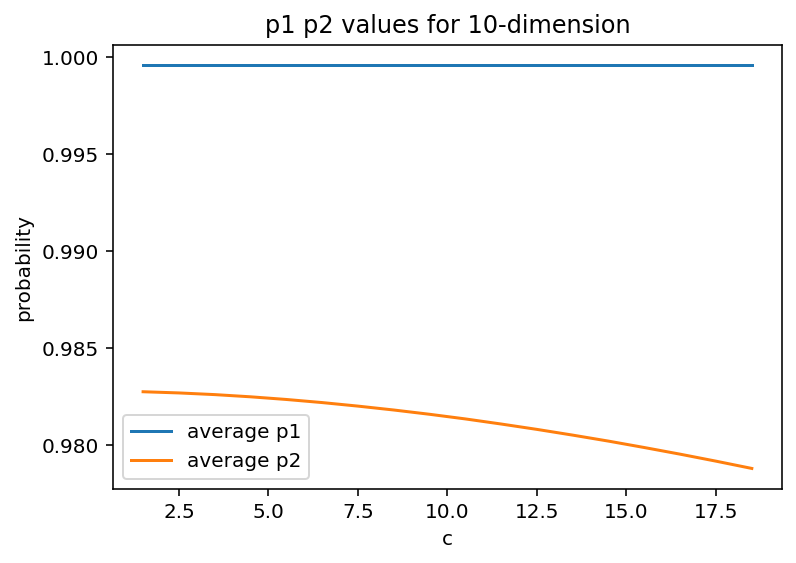}
    \includegraphics[width=0.48\linewidth]{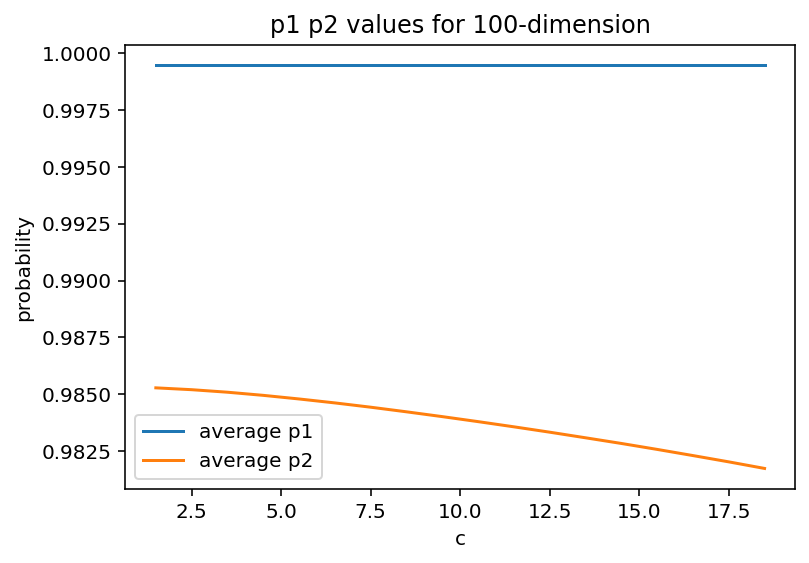}
    \includegraphics[width=0.48\linewidth]{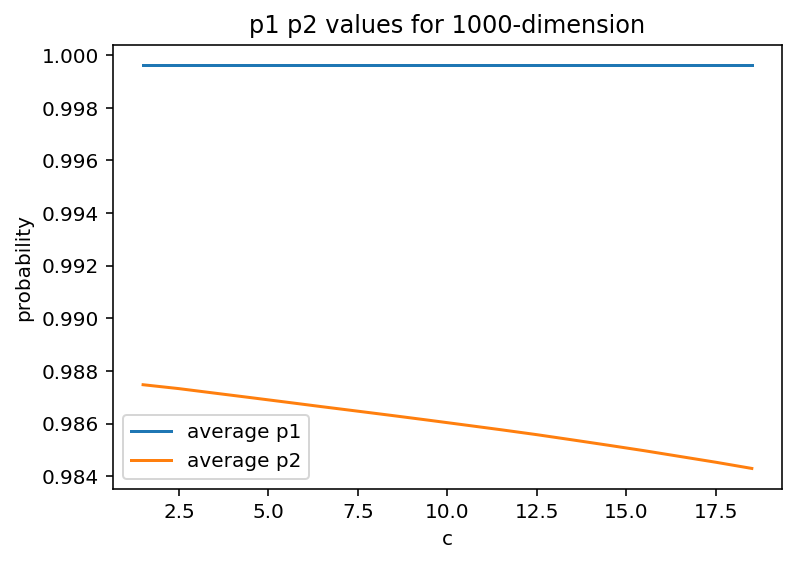}
    \caption{Average $p_1, p_2$ of Hyperbolic LSH with varying $c$ for different dimensions ($y$-values do not start at 0)}
    \label{fig:p1p2}
\end{figure}

\begin{figure}[h!]
    \centering
    \includegraphics[width=0.48\linewidth]{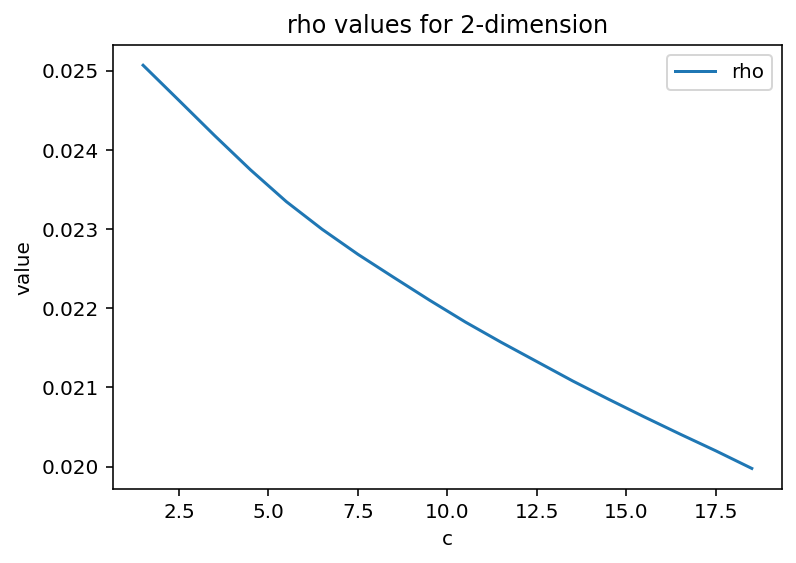}
    \includegraphics[width=0.48\linewidth]{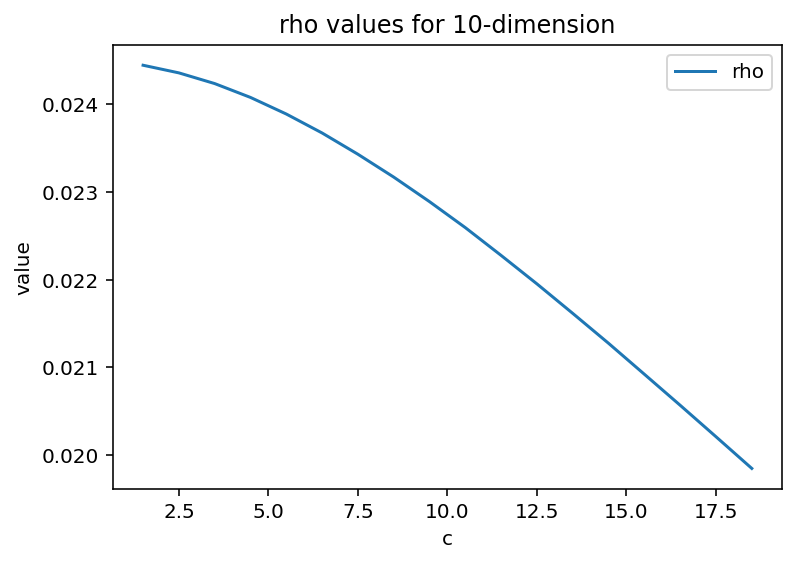}
    \includegraphics[width=0.48\linewidth]{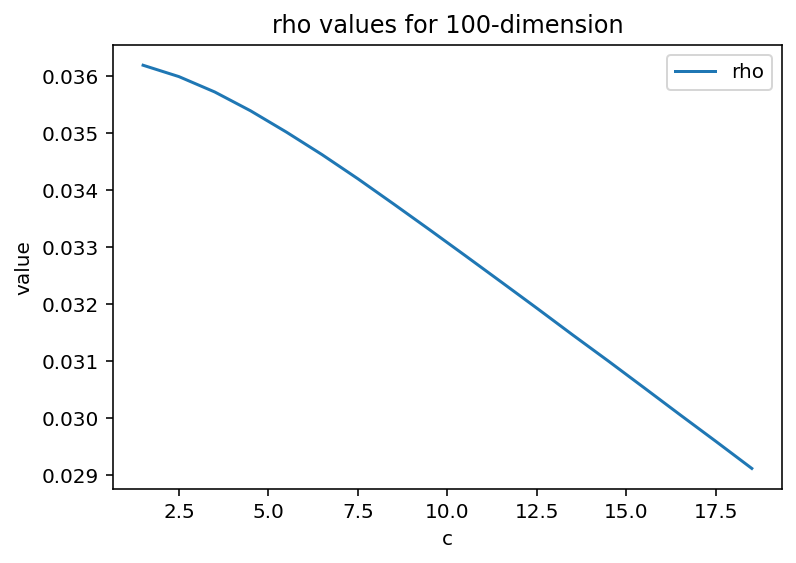}
    \includegraphics[width=0.48\linewidth]{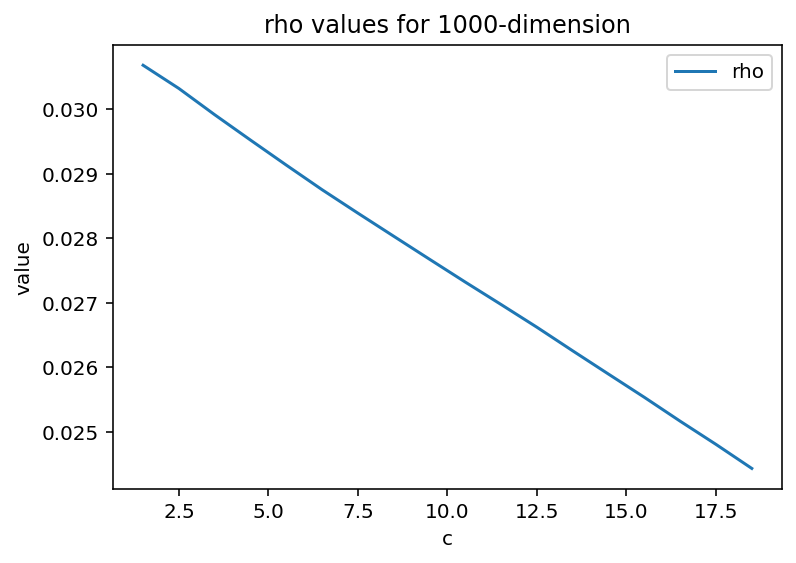}
    \caption{Average $\rho$ of Hyperbolic LSH with varying $c$ for different dimensions ($y$-values do not start at 0)}
    \label{fig:rho}
\end{figure}

\begin{figure}[h!]
    \centering
    \includegraphics[width=0.48\linewidth]{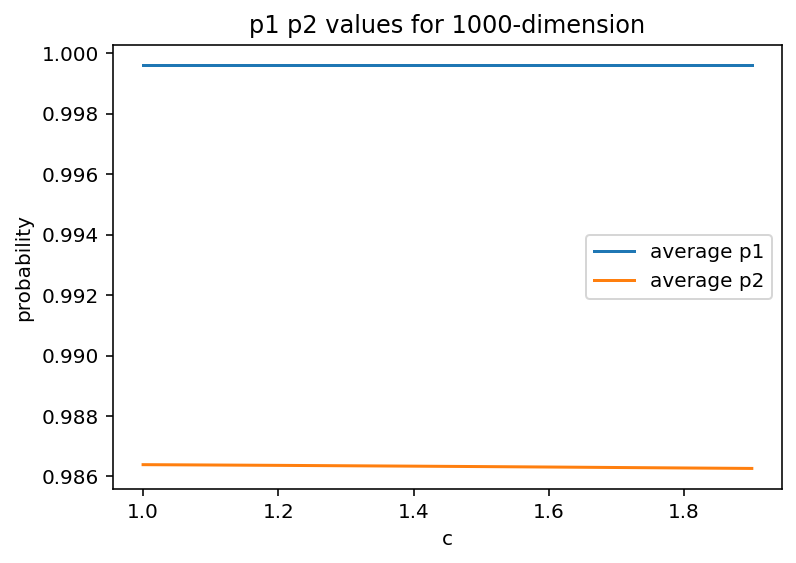}
    \includegraphics[width=0.48\linewidth]{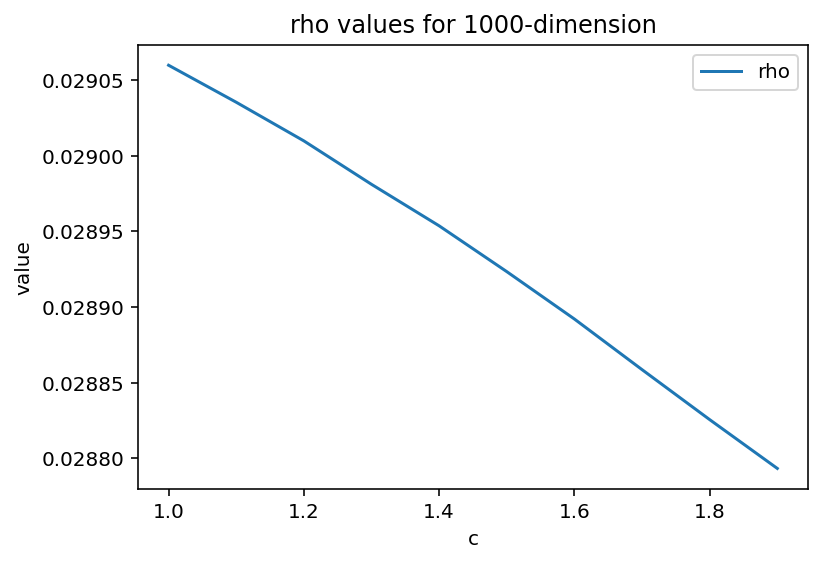}
    \caption{Average $\rho$ and average $p_1, p_2$ in high-dimension when $c$ is small ($y$-values do not start at 0)}
    \label{fig:low-c}
\end{figure}

\begin{figure}[h!]
    \centering
    \includegraphics[width=0.48\linewidth]{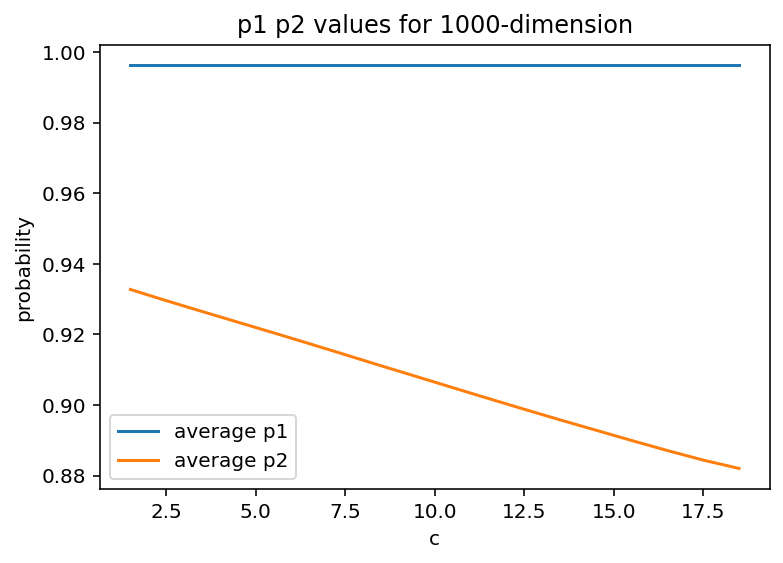}
    \includegraphics[width=0.48\linewidth]{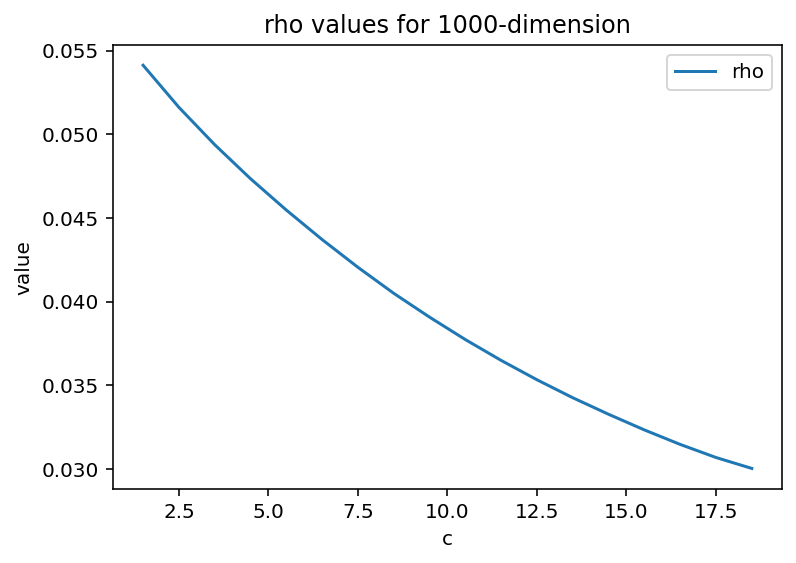}
    \caption{Average $\rho$ and average $p_1, p_2$ in high dimension and $R = .9$ in Euclidean coordinates ($y$-values do not start at 0)}
    \label{fig:boundary}
\end{figure}

\begin{figure}[h!]
    \centering
    \includegraphics[width=0.48\linewidth]{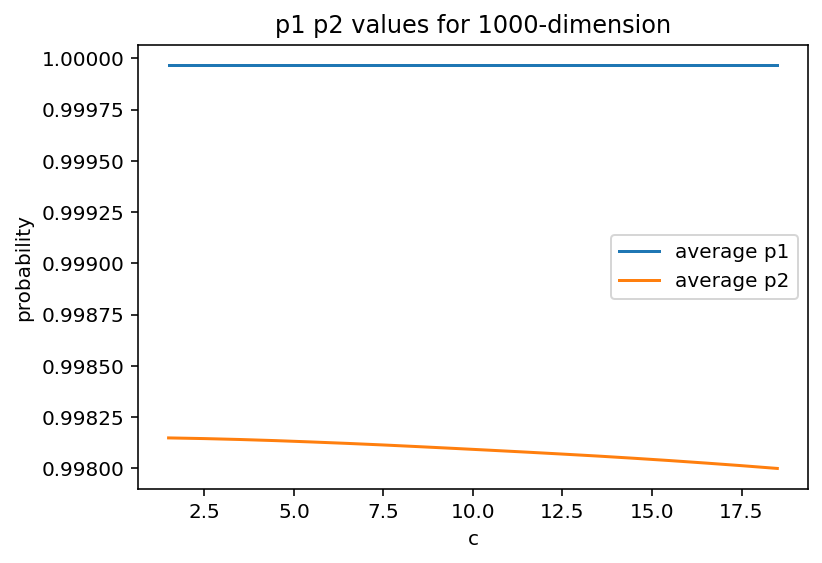}
    \includegraphics[width=0.48\linewidth]{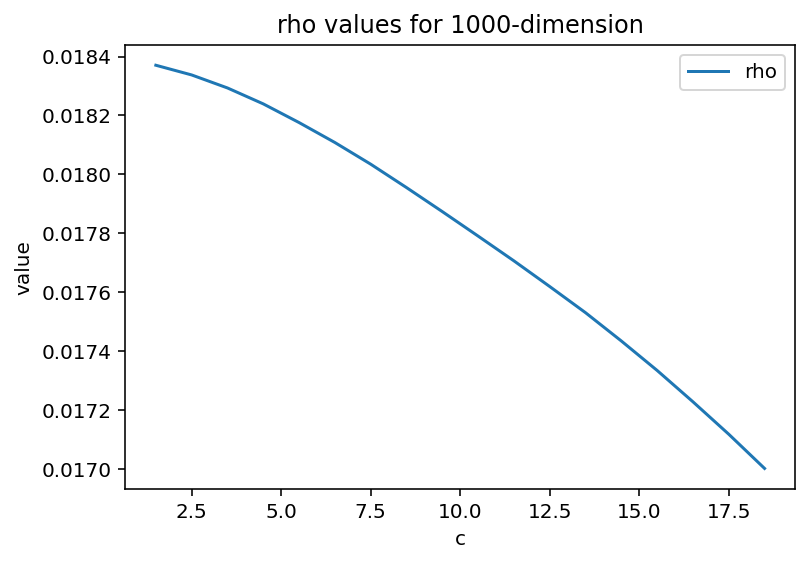}
    \caption{Average $\rho$ and average $p_1, p_2$ in high dimension and $R = .999$ in Euclidean coordinates ($y$-values do not start at 0)}
    \label{fig:boundary}
\end{figure}

\begin{figure}[h!]
    \centering
    \includegraphics[width=0.48\linewidth]{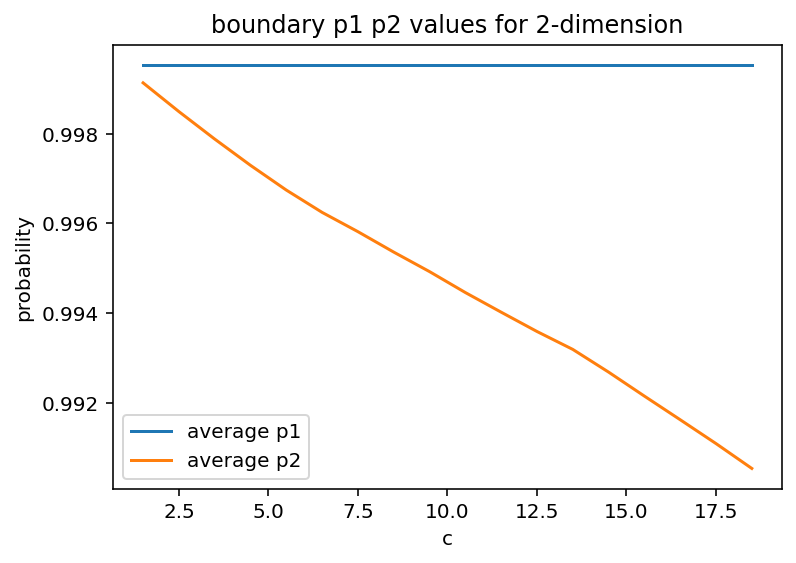}
    \includegraphics[width=0.48\linewidth]{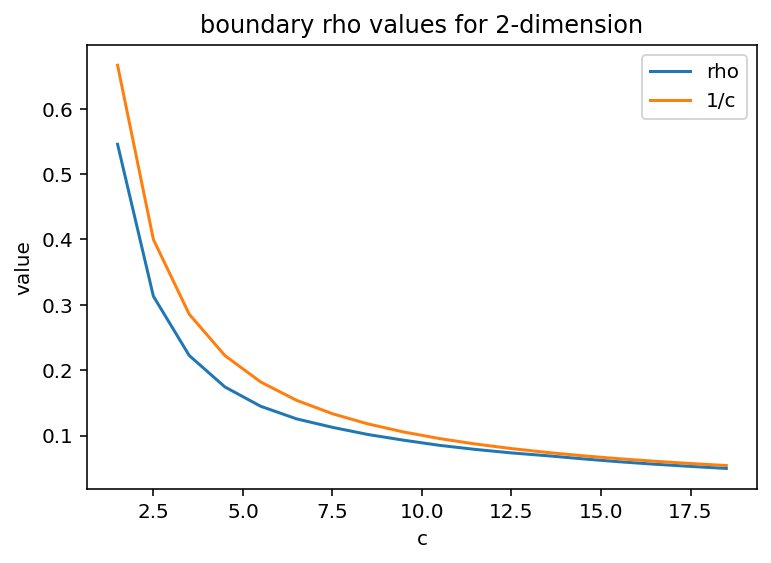}
    \caption{Average $\rho$ and average $p_1, p_2$ for only pairs at distance approximately $r$ or $cr$ ($y$-values do not start at 0)}
    \label{fig:boundary}
\end{figure}
We see that as expected, $p_1$ and $p_2$ are close to 1 in all cases. Due to the randomness of the data, the probabilities vary slightly in percentage. Since $\rho$ is related to the complements of these probabilities, varying slightly in proportion when so close to 1 causes the significant differences in $\rho$ values that we are observing. However, in all cases, the $\rho$ graphs are below $1/c$. We note that they are not exactly $1/c$ because we considered all pairs of points with distance below $r$ or above $cr$, not exactly equal to $r$ or $cr$, so $1/c$ should be an upper bound of what we obtain.

Next, we use the same data for $d = 2$, but vary $c = 1 + .1 k$ for $k = 1, ... 10$ to see if the LSH still works for small values of $c$ even though our proof only has guarantees for $c \geq 1.59$. From \Cref{fig:low-c} we see fairly similar results, just $p_2$ varying less, but LSH still achieving $\rho$ lower than $\frac{1}{c}$.

We also generated some data for $d = 1000$, where we vary $R$ so that the points lie in a circle of Euclidean radius $.9$ and $.999$. We see that $p_1, p_2$ and $\rho$ change a little, but not significantly and still fall well within the theoretical bounds.

Lastly, we generate similar data for $d = 2$, but containing 2000 points. Then, we analyze $p_1$ and $p_2$ for pairs of points that are at the boundary of acceptable distances. Precisely, we choose pairs of points that have distance within $[.9r, r]$ or $[cr, 1.1cr]$. This gives us an approximately linear graph for $p_2$ and $\rho$ curve that approximates $\frac{1}{c}$.

The experiment setup with a fixed $r$ and a reasonably large constant $c$ (rather than $1+\eps$ for a very small $\eps$) is motivated by real world applications such as similarity learning which focuses on differentiating data elements that are very similar or sufficiently far away. Many data modalities and data tasks fall into this category. For example, in image recognition the goal is to recognize two pictures with similar scenes or semantics. But for two pictures on completely different subjects, it is sufficient to know that they are far away and the exact values of dissimilarities are less important. One important take-away message from our experiments is that the empirical value of $\rho$ with the synthetic data set is very small with a mild dependency on dimension and decreasing value of $c$.  It is much smaller than our theoretical upper bound -- partly because we consider all pairs whose distance are often much smaller than $r$ or much greater than $cr$. Since the storage requirement for approximate nearest neighbor data structure built with LSH grows in terms of $O(n^{1+\rho})$ and the query cost grows in terms of $O(n^{\rho})$. The empirical small value of $\rho$  suggests great potential for practical use.

%\jie{curious about small c<1.5.}

% 
\section{Conclusion}

 To our knowledge, this is the first work for locality sensitive hashing in hyperbolic space. Our construction achieves a performance parameter $\rho=O(1/c)$ for general hyperbolic space. To complement, we show there is a lower bound of $\rho \approx \Omega(1/c^2)$. The major open question is whether there can be improvement with a better $\rho$ or show a stronger lower bound for the LSH of hyperbolic distances. Note that prior results are obtained from Boolean analysis that give lower bounds on the Hamming space, which implies a lower bound for $\ell_p$ spaces. We believe novel techniques are required to show a non-trivial lower bound for hyperbolic space. Another open direction is data-dependent LSH, which outperforms classical LSH in Euclidean and Hamming space. It would be interesting to see if this result extends to hyperbolic space.

\bibliography{reference,pq-bib}

\clearpage

\appendix
\section{Proof of Dimension Reduction in $\mathbb{H}^d$}\label{sec:JLproof}
For completeness, we include the full proof of \Cref{lem:F}.

\begin{proof}[Proof of \Cref{lem:F}]
We first prove case 1. Take $\gamma=1+\delta$. By the mean value theorem, there is a value $r\leq \hat{r}\leq (1+\delta)r$ such that \[F_{z_1, z_2}((1+\delta)r)=F_{z_1, z_2}(r)+ \frac{dF_{z_1, z_2}(\hat{r})}{dr}\cdot \delta r. \]
Now we upper bound the derivative of $F_{z_1, z_2}$:
\begin{align*}
\frac{dF_{z_1, z_2}(\hat{r})}{dr} &= \frac{2\hat{r}}{\sqrt{\hat{r}^2+|z_1-z_2|^2}\sqrt{\hat{r}^2+|z_1-z_2|^2+4z_1z_2}}\\
&\leq \frac{2}{\sqrt{\hat{r}^2+|z_1-z_2|^2+4z_1z_2}} \tag{$\hat{r}^2+|z_1-z_2|^2+4z_1z_2\geq \hat{r}^2$}\\
&\leq \frac{2}{\sqrt{r^2+|z_1-z_2|^2+4z_1z_2}} \tag{$\hat{r}\geq r$}
\end{align*}
Denote by $\Delta=F_{z_1, z_2}(r)$. Re-arrange \Cref{eq:F}, we have
\begin{equation}\label{eq:cosh}
    r^2+(z_1-z_2)^2=2z_1z_2(\cosh \Delta -1).
\end{equation}
Now we can bound
\begin{align*}
\frac{dF_{z_1, z_2}(\hat{r})}{dr} \cdot \delta r&\leq \frac{2\delta r}{\sqrt{2z_1z_2(\cosh \Delta +1)}} \tag{Plug in \Cref{eq:cosh}}\\
&\leq 2\delta \sqrt{\frac{r^2+(z_1-z_2)^2}{2z_1z_2}}\cdot \frac{1}{\sqrt{\cosh \Delta +1}} \tag{$r\leq \sqrt{r^2+(z_1-z_2)^2}$}\\
&=2\delta \sqrt{\frac{\cosh \Delta -1}{\cosh \Delta +1}}\tag{Plug in \Cref{eq:cosh}}\\
& =2\delta \tanh{\frac{\Delta}{2}}\\
&\leq \delta \Delta
\end{align*}
The last inequality applies the inequality $\tanh t \leq t$, when $t>0$. This is because of the following formula for $\tanh x$:
\[\tanh x = \frac{1}{\frac{1}{x}+\frac{1}{\frac{3}{x}+\frac{1}{\frac{5}{x}+\cdots}}}.\]
Combining everything, we have $F_{z_1, z_2}(\gamma \cdot r)\leq \gamma \cdot F_{z_1, z_2}(r)$, when $\gamma\geq 1$.

For case 2, take $\gamma=1/\gamma'$, $\gamma\geq 1$. Thus \[F_{z_1, z_2}(r)=F_{z_1, z_2}(\gamma \cdot \gamma' r)\leq \gamma \cdot F_{z_1, z_2}(\gamma'\cdot r)=\frac{1}{\gamma'} \cdot F_{z_1, z_2}(\gamma'\cdot r).\]
In other words, $F_{z_1, z_2}(\gamma' \cdot r)\geq \gamma' \cdot F_{z_1, z_2}(r)$.
\end{proof}

\section{Proof of LSH in $\mathbb{H}^2$}\label{appendix:proof-H2}

%\subsection{Proof of \Cref{lem:collision-prob-hr}}

%A technical caveat for the HR-LSH algorithm is that we require a well-defined measure on the family of all geodesics intersecting a hyperbolic ball of radius $R$. Such a measure is essential to ensure that random sampling of geodesics is meaningful and invariant under hyperbolic isometries. To this end, we appeal to a classical tool from integral geometry, namely \emph{Crofton’s formula}. It provides a profound connection between geometric quantities (such as lengths, areas, or volumes) and averages over the space of geodesics. 

\begin{proof}[Proof of Claim \ref{clm:integral}]
We begin by evaluating the inner integral with respect to $t$. The integral of $\cosh(t)$ is $\sinh(t)$.

\begin{align*}
    \int_0^{\text{arctanh}(\tanh(r/2)\cos\theta)} \cosh(t) \,dt &= \left[ \sinh(t) \right]_0^{\text{arctanh}(\tanh(r/2)\cos\theta)} \\
    &= \sinh(\text{arctanh}(\tanh(r/2)\cos\theta)) - \sinh(0) \\
    &= \sinh(\text{arctanh}(\tanh(r/2)\cos\theta))
\end{align*}

To simplify this expression, we use the identity $\sinh(\text{arctanh}(x)) = \frac{x}{\sqrt{1-x^2}}$. Let $x = \tanh(r/2)\cos\theta$. The expression becomes:
\[ \frac{\tanh(r/2)\cos\theta}{\sqrt{1 - \tanh^2(r/2)\cos^2\theta}} \]
Now, we substitute this result into the outer integral:
\[ I = \int_{-\pi/2}^{\pi/2} \frac{\tanh(r/2)\cos\theta}{\sqrt{1 - \tanh^2(r/2)\cos^2\theta}} \,d\theta \]
Let the constant $a = \tanh(r/2)$. The integrand $f(\theta) = \frac{a \cos\theta}{\sqrt{1 - a^2\cos^2\theta}}$ is an even function, since $\cos(-\theta) = \cos(\theta)$. Therefore, the integral over the symmetric interval $[-\pi/2, \pi/2]$ is twice the integral over $[0, \pi/2]$:
\[ I = 2 \int_{0}^{\pi/2} \frac{a \cos\theta}{\sqrt{1 - a^2\cos^2\theta}} \,d\theta \]
Using the identity $\cos^2\theta = 1 - \sin^2\theta$, we can rewrite the denominator:
\[ I = 2 \int_{0}^{\pi/2} \frac{a \cos\theta}{\sqrt{1 - a^2(1-\sin^2\theta)}} \,d\theta = 2 \int_{0}^{\pi/2} \frac{a \cos\theta}{\sqrt{1 - a^2 + a^2\sin^2\theta}} \,d\theta \]
We perform a substitution with $u = a\sin\theta$, so $du = a\cos\theta d\theta$. The limits of integration for $u$ are from $a\sin(0)=0$ to $a\sin(\pi/2)=a$.
\[ I = 2 \int_0^a \frac{du}{\sqrt{(1-a^2) + u^2}} \]
This is a standard integral form. Evaluating it yields:
\[ I = 2 \left[ \text{arsinh}\left(\frac{u}{\sqrt{1-a^2}}\right) \right]_0^a = 2 \, \text{arsinh}\left(\frac{a}{\sqrt{1-a^2}}\right) \]
Now, we substitute $a = \tanh(r/2)$ back into the expression. The argument of the $\text{arsinh}$ function simplifies as follows, using the identity $1 - \tanh^2(x) = \text{sech}^2(x)$:

\[ \frac{a}{\sqrt{1-a^2}} = \frac{\tanh(r/2)}{\sqrt{1-\tanh^2(r/2)}} = \frac{\tanh(r/2)}{\sqrt{\text{sech}^2(r/2)}} = \frac{\tanh(r/2)}{\text{sech}(r/2)} \]

By expressing $\tanh$ and $\text{sech}$ in terms of $\sinh$ and $\cosh$, we get:
\[ \frac{\sinh(r/2)/\cosh(r/2)}{1/\cosh(r/2)} = \sinh(r/2) \]
Substituting this simplified result back into the equation for $I$:
\[ I = 2 \, \text{arsinh}(\sinh(r/2)) = r \]
This completes the proof.
\end{proof}

\end{document}